\documentclass[aps,pra, two column,10pt,superscriptaddress,nofootinbib, notitlepage]{revtex4-2}

\usepackage{pgfplots}
\pgfplotsset{compat=1.18}
\usepackage{amsmath}
\usetikzlibrary{plotmarks}
\usepackage{eufrak} 
\usepackage{physics}
\usepackage{amsmath,amssymb,amsthm,amstext,mathrsfs}
\usepackage{amsthm}
\usepackage[colorlinks=true,citecolor=blue,urlcolor=blue]{hyperref}
\usepackage{comment}
\usepackage{eufrak} %
\usepackage{enumerate}
\usepackage{times,txfonts}
\usepackage{braket}
\usepackage{svg}
\usepackage{color}
\usepackage{natbib}
\usepackage{amsmath,blkarray}
\usepackage{mathtools}
\usepackage{latexsym}
\usepackage{tabularx, booktabs}
\usepackage{graphics,epstopdf}
\usepackage{graphicx, lipsum}
\usepackage{amsfonts}
\usepackage{enumerate}
\usepackage{color,soul}

\newcommand{\be}{\begin{equation}}
\newcommand{\ee}{\end{equation}}
\newcommand{\ba}{\begin{eqnarray}}
\newcommand{\ea}{\end{eqnarray}}

\newtheorem{thm}{Theorem}

\definecolor{ss}{RGB}{250,80,220}
\begin{document}
\title{Semi-device-independent self-testing of unitary operations}
\author{Rajdeep Paul}\email{rajdeeppaul22@gmail.com}\affiliation{Department of Physics, Indian Institute of Technology Hyderabad, Telangana-502284, India }
		\author{Prabuddha Roy}
        \email{prabuddhar@kias.re.kr}
         \affiliation{Quantum Universe Center, Korea Institute for Advanced Study, Seoul 02455, Republic of Korea }
         \author{A. K. Pan}
	\email{akp@phy.iith.ac.in}
	 \affiliation{Department of Physics, Indian Institute of Technology Hyderabad, Telangana-502284, India }
\begin{abstract}
We present a hitherto unexplored semi-device-independent (SDI) self-testing protocol designed to certify unitary operations within a variant of prepare-measure framework. We consider a communication game which we refer to as a variant of $3$-bit prepare-measure random access code (PMRAC) involving two parties, Alice and Bob, who share a prior two-qubit quantum state. Alice encodes her message by applying unitary operations on her subsystem and sends it to Bob. To decode the message, Bob performs a measurement on the whole system. We demonstrate that the optimal quantum advantage of the variant of  $3$-bit PMRAC over the classical bound enables the self-testing of Alice's unitary operations and Bob's measurements. The derivation of the optimal quantum success probability is fully analytical. The approach is so elegant that it can be generalized for any arbitrary $n$-bit PMRAC and may also be extended to other prepare-measure communication games. 

\end{abstract}
\maketitle
\emph{Introduction:-} Self-testing refers to the purest form of quantum certification protocol in which the inner working of devices remains uncharacterized and the dimension of the quantum system is unknown. The optimal quantum violation of a Bell inequality \cite{Bell1964,Brunner2014} enables such a device-independent (DI) self-testing. Since the inception of the first DI self-testing protocol \cite{mayers1998,mayers2004s}, a plethora of protocols have been demonstrated. For recent work, we refer the reader to a review \cite{Supic2020rev}.

However, DI certification encounters practical challenges due to the necessity of a loophole-free Bell test. Although progress has been made in the implementation of such tests \cite{Hensen2015,Shalm2015,Giustina2015}, it remains a challenging task. Due to this, SDI self-testing in prepare-measure scenarios \cite{Tavakoli2018, Tavakoli2020sdic,supic2020input,Pauwels2022,Navascues2023} has become a promising alternative. Such protocols usually posit an upper bound on the Hilbert space dimension, but the inner workings of the devices remain uncharacterized. Of late, a multitude of SDI protocols have been proposed to self-test pure states and projective measurements \cite{Tavakoli2018,Tavakoli2020sdic}, non-projective measurements \cite{Mironowicz2019,Smania2020}, mutually unbiased bases \cite{Farkas2019}, randomness \cite{Li2011,Li2012,Lunghi2015} and unsharpness parameters \cite{Mohan2019,Miklin2020,Abhyoudai2023}. Experimental demonstrations of SDI self-testing of states and measurements have also been reported \cite{Smania2020,Anwer2020,Tavakoli2020sc,Zhang2020}. The certification of quantum operations has been demonstrated \cite{Dallarno2017, SekatskiPrl2018, Wagner2020, Sarkar2024}, in particular, the self-testing of quantum gate has been proposed \cite{Magniez2006, sarkar2025gate}. The classical verifications of quantum operations have been demonstrated in \cite{Reichardt2013, Noller2025}. 

In this Letter, we demonstrate a hitherto unexplored SDI self-testing protocol to certify the unitary operations in a variant of prepare-measure scenario. Specifically, we consider a prepare-measure communication game between two parties, widely known as the random access code (RAC). In traditional PMRAC, the sender Alice prepares the states and then sends them to the receiver Bob. In contrast, in our variant of quantum PMRAC, Alice and Bob initially share a quantum state. Alice then performs local quantum operations on her subsystem of the shared bipartite state of known dimension to encode the information corresponding to the input she receives and sends her subsystem to Bob. According to the input Bob receives, he performs a suitable measurement on the joint system to win the game with the highest success probability. 

For our purpose, we consider the two-level $3$-bit PMRAC and assumue that Alice and Bob share a two-qubit state. We demonstrate that the optimal quantum success probability exceeds the classical $3$-bit RAC. We demonstrate that the optimal quantum value enables the SDI self-testing of shared two-qubit maximally entangled state, Bob's measurements, and Alice's unitary operations. The analytical technique developed here can be straightforwardly adopted for arbitrary $n$-bit PMRAC in the variant of prepare-measure scenario. 


\emph{ $3$-bit classical RAC:-} We briefly recapitulate the two-level $3$-bit classical RAC \cite{Ambainis2009,Tavakoli2015, Roy2024}. In a $3$-bit RAC, Alice randomly receives an input string $x\in \{0,1\}^3$, encodes this information into a physical system by a preparation procedure $P_{x}$ and sends it to Bob. Bob randomly receives input $y\in[n]$ and performs a measurement $B_{y}\equiv\{E_y^b\}$ that produces outcome $b \in \{0,1\}$. The winning rule of the RAC is $b=x_y$, i.e., Bob needs to output $y^{\text{th}}$ bit of Alice's input.  The average success probability can be written as
	\begin{align}
		\label{SP}
		(\mathcal{S}_C)^{3\rightarrow m}=\dfrac{1}{24}\sum_{y=1}^{3}\sum_{x\in \{0,1\}^{3}}p(b=x_{y}|P_{x},E_{y}^{b})
	\end{align}

 To help Bob, Alice can communicate $m<3$ bit message to Bob. The higher the value of $m$, the higher the success probability can be obtained. If Alice communicates one bit of information, then the maximum success probability is $(\mathcal{S}_C)^{3\rightarrow 1}=3/4$ \cite{Ambainis2009}. However, if Alice sends two bits, then $(\mathcal{S}_C)^{3\rightarrow 2}=5/6> (\mathcal{S}_C)^{3\rightarrow 1}$ \cite{Roy2024}.

\emph{ 3-bit quantum RAC:-} In quantum theory, RAC tasks are commonly played in two scenarios. (i) Prepare-measure scenario \cite{ Hayashi2006, Ambainis2009, Spekkens2009,  Chailloux2016, Hameedi2017b,Tavakoli2018, Mohan2019,   Pauwels2022, Abhyoudai2023,singh2025}: Alice prepares the quantum states $\rho_x$ to encode her bit string and send the system to Bob through a classical channel. (ii) Entanglement-assisted scenario \cite{Pawlowski2010,Muhammad2014,Banik2015,Tavakoli2016,Hameedi2017a,Laplante2018,Tavakoli2018b,Ghorai2018,Pan2020,Pan2021,Roy2023,Paul2024}: Alice and Bob pre-share an entangled state, and by performing the measurement on her subsystem, Alice steers the states $\rho_x$ to Bob and sends $m$-bit of the classical message. However, there is yet another scenario that has recently been introduced Tavakoli \emph{et al.} \cite{TavakoliPRXQ2021}, motivated by the dense coding protocol \cite{Bennett1992}. (iii) Entangled-assisted prepare-measure scenario: This scenario features a hybrid prepare-measure scenario \cite{Piveteau2022} assisted by pre-shared entanglement between Alice and Bob, and  Alice applies unitary operations on her system and sends it to Bob.  In this work, we consider a variant of PMRAC which is quite close to the entangled-assisted prepare-measure scenario with a difference that we do not consider prior entanglement being shared between Alice and Bob. Rather, we show that the optimal quantum success probability of our variant of PMRAC self-tests the initial shared state to be maximally entangled. In such a variant of PMRAC we demonstrate the self-testing of Alice's local unitary operations.
\begin{figure}[h]
                 \centering
                 \includegraphics[width=8cm, height=3.5cm]{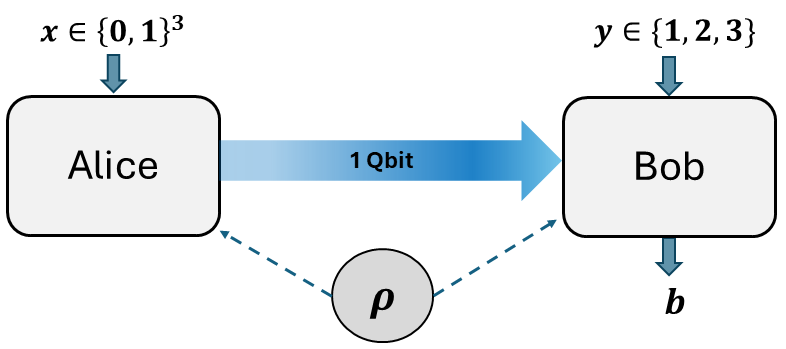}
                 \caption{\textbf{A variant of quantum PMRAC:} Alice and Bob share a two-qubit state $\rho$. Alice performs quantum channel operation $\Lambda_x$ on her qubit as per the input $x\in\{0,1\}^3$ she receives and transmits it to Bob, who performs projective measurement $\Pi_{y}^{b}$ on both qubits and produce binary outcome.}\label{eapmrac}
                 \end{figure}

\emph{$3$-bit variant of quantum PMRAC:} The encoding and decoding strategies in our variant of PMRAC are as follows. Alice and Bob share a bipartite state $\rho \in \mathbb{C}^{2} \otimes \mathbb{C}^{2}$. Upon receiving the input $x\in \{0,1\}^3$, Alice applies a quantum operation $\Lambda_x$ (a completely positive trace-preserving map) to her subsystem, as depicted in Fig.\ref{eapmrac}. She then sends her qubit to Bob. As a result, Bob holds one of the eight two-qubit states, $ 
\rho_{x}=\Lambda_{x}(\rho) $ and performs a two-outcome projective measurement $\{\Pi_{y}^b\}$ on the two-qubit system. The average quantum success probability of the $3$-bit variant of PMRAC can then be written as
\begin{eqnarray}\label{QSP}
  \mathcal{S}_{Q}=\dfrac{1}{24}\sum_{y=1}^{3}\sum_{x\in \{0,1\}^{3}}\Tr\hspace{2pt}[\rho_{x}\Pi^b_y]
\end{eqnarray}

 We introduce an elegant analytical technique to derive the optimal success probability. Using Eq. (\ref{QSP}), the quantum success probability can be explicitly written as
\begin{eqnarray}
\label{eqnA5m}
\nonumber
\mathcal{S}_{Q}=\dfrac{1}{2}&+&\dfrac{1}{48}\mathrm{Tr}\Big[(\rho_{000}+\rho_{011}-\rho_{101}-\rho_{110})B_{1}\\
 \nonumber
 &+&(\rho_{000}-\rho_{011}+\rho_{101}-\rho_{110})B_{2}\\
  \nonumber
 &+&(\rho_{000}-\rho_{011}-\rho_{101}+\rho_{110})B_{3}\\
 \nonumber
 &+&(\rho_{001}+\rho_{010}-\rho_{100}-\rho_{111})B_{1}\\
 &+&(\rho_{001}-\rho_{010}+\rho_{100}-\rho_{111})B_{2}\\
 \nonumber
 &+&(-\rho_{001}+\rho_{010}+\rho_{100}-\rho_{111})B_{3}\Big]
\end{eqnarray} 
where we used $\Pi_{y}^{b} = (\openone_2+(-1)^{b}B_{y})/{2}$. 

Let us look at the first term in Eq. (\ref{eqnA5m}). To obtain the maximum value, the states $\rho_{000}, \rho_{011}, \rho_{110}$ and $\rho_{101}$ must be eigenstates of $B_1$. This demands that they are pure states and satisfy $\rho_{000}+\rho_{110}+\rho_{011}+\rho_{101}=\openone_2\otimes \openone_2$. Therefore, the initial two-qubit state $\rho$ has to be a pure state and the quantum operation $\Lambda_{x}$ Alice applies to her subsystem has to be a unitary operation $U_{x}$, otherwise the purity of the state cannot be preserved. A similar argument holds for each of such terms in Eq. (\ref{eqnA5m}). This leads us to define three two-qubit observables of the form  
\begin{eqnarray}\label{tila1}
\nonumber
M_1 &=& -\rho_{000}-\rho_{011}+\rho_{101}+\rho_{110}\\
M_2&=& -\rho_{000}+\rho_{011}-\rho_{101}+\rho_{110}\\
\nonumber
M_3&=&-\rho_{000}+\rho_{011}+\rho_{101}-\rho_{110}
\end{eqnarray}
Note that, by construction, $[M_{1},M_{2}]=[M_{2},M_{3}]=[M_{1},M_{3}]=0$ and $M_1=-M_2 M_3$.

Similarly, the other set of states $\rho_{001}, \rho_{010},\rho_{111}$ and $\rho_{100}$ are also mutually orthogonal pure states, which again implies $\rho_{001}+\rho_{010}+\rho_{111}+\rho_{100}=\openone_2\otimes \openone_2$. Thus, we can construct three more two-qubit observables given by
\begin{eqnarray}\label{tila2}
\nonumber
N_1&=& \rho_{001}+\rho_{010}-\rho_{100}-\rho_{111}\\
N_2&=&\rho_{001}-\rho_{010}+\rho_{100}-\rho_{111}\\
\nonumber
N_3&=&-\rho_{001}+\rho_{010}+\rho_{100}-\rho_{111}
\end{eqnarray}
where the observables satisfy $[N_{1},N_{2}]=[N_{2},N_{3}]=[N_{1},N_{3}]=0$ and $N_1=-N_2 N_3$.

Thus, the success probability in Eq. (\ref{eqnA5m}) takes the form 
\begin{equation}\label{sp33}
\mathcal{S}_{Q}=\dfrac{1}{2}+\dfrac{\Delta}{48}
\end{equation}
where the correlation function 
\begin{eqnarray}
\label{delta33}
\Delta=\sum_{y=1}^3\Tr[(N_y-M_y) B_y]
\end{eqnarray}

In order to find the maximum quantum value of $\Delta$, we write the correlation function as
\begin{eqnarray}
\Delta=\Tr[\sum\limits_{y=1}^{3}\omega_y\mathscr{M}_y B_y ]\label{wab}
\end{eqnarray}
where $\mathscr{M}_y$ and $\omega_y$ are defined as the following

\begin{eqnarray}
     &&\mathscr{M}_y = \frac{N_y-M_y}{\omega_y} \ ; \ \ 
    \omega_y=||N_y-M_y||; \ \ \forall y\in[3] \label{omegai}
\end{eqnarray}
  such that $||\mathscr{M}_y||=1, \forall y\in[3]$. Here, $||\cdot||$ is the scaled Frobenious norm, given by $||\mathcal{O}||=\frac{1}{\sqrt{4}}\sqrt{Tr[\mathcal{O}^{\dagger}\mathcal{O}]}$.

  Interestingly, since $\mathscr{M}_y$ and $ B_y$ are dichotomic normalized observables, the maximum value of $\mathscr{M}_y B_y$ occurs only when
$\mathscr{M}_y=B_y; \ \ \forall y\in[3]$. This further implies
\begin{equation}\label{optbn}
\Delta_{opt}= \Tr[\max\left(\sum\limits_{y=1}^{3}\omega_{y} \right)\openone_2\otimes\openone_2]=4\max\left(\sum\limits_{y=1}^{3}\omega_{y} \right)
\end{equation}

Next, using the inequality $\sum\limits_{y=1}^{3}\omega_{y}\leq \sqrt{n \sum\limits_{y=1}^{3} (\omega_{y})^{2}}$ where equality holds only when $\omega_y=\omega_{y^\prime}, \forall y\neq {y^\prime}\in [3]$, from Eq.~(\ref{optbn}), we get
\begin{eqnarray}\label{A2}
\Delta&\leq&4\qty( \max \sqrt{3 \ \qty(\omega_1^2+\omega_2^2+\omega_3^2)} \ )
\end{eqnarray}
Now, by writing $\omega_y$ from Eq.~(\ref{omegai}) as
\begin{eqnarray}
\omega_y&=&\frac{1}{2}\sqrt{\Tr[2\openone_4-2 M_y N_y]}
\end{eqnarray}
and by putting them in Eq.~(\ref{A2}) we get
\begin{eqnarray}\label{A1}
\Delta
&\leq&4\max\sqrt{3 \ \Big[6-\qty(\Tr[M_1  N_1 + M_2 N_2 + M_3 N_3])/2\Big]}
\end{eqnarray}
By rewriting $M_{y}N_{y}, \ \forall y \in [3]$ in terms of $\rho_{x},\forall x\in\{0,1\}^3$ we find
\begin{eqnarray}
    \Tr[M_1  N_1 + M_2 N_2 + M_3 N_3]&=&-4+4(Tr[\rho_{000} \ \rho_{111} +\rho_{011} \ \rho_{100}\nonumber\\
&&+\rho_{101} \ \rho_{010}+\rho_{110} \ \rho_{001}])
\end{eqnarray}
Therefore,
\begin{eqnarray} \label{delta2}
\nonumber
\Delta\leq4&& \ \max \Big[3(8-2(Tr[\rho_{000} \ \rho_{111} +\rho_{011} \ \rho_{100}+\rho_{101} \ \rho_{010}\nonumber\\
&&+\rho_{110} \ \rho_{001}])\Big]^\frac{1}{2}
\end{eqnarray}
which finally leads to 
\begin{equation}
   \Delta \leq 8\sqrt{6} \label{optbn1}
\end{equation} when \ba
\label{ovzero}\nonumber
\Tr[\rho_{000}\rho_{111}]=\Tr[\rho_{011}\rho_{100}]=\Tr[\rho_{101} \ \rho_{010}]=\Tr[\rho_{110} \ \rho_{001}]=0\\\ea

This also implies 
\begin{equation}
\label{mn}
    \Tr[M_1 N_1]+\Tr[M_2 N_2]+\Tr[M_2 N_2]=-4
\end{equation}
Using Eq. (\ref{mn}), we proved in Appendix \ref{appsos} that optimization fixes $\omega_y=\sqrt{\frac{8}{3}}, \forall y\in [3]$ and, therefore, $\Delta_{opt}=8\sqrt{6}$. Hence, the optimal success probability from Eq.~(\ref{sp33}) is derived as
\begin{equation}\label{opts3}
\mathcal{S}_Q^{opt}=\dfrac{1}{2}+\frac{1}{\sqrt{6}} \approx 0.908
\end{equation}
Note that $\mathcal{S}_Q^{opt}>(\mathcal{S}_C)^{3\rightarrow 2}$ i.e., optimal quantum success probability outperforms the classical RACs. We remark here that the optimal quantum success probability of traditional $3$-bit PMRAC by sending one qubit was derived as $\mathcal{S}_{QRAC}=\frac{1}{2}\qty(1+\frac{1}{\sqrt{3}}) \approx 0.785$~\cite{Ambainis2009,Spekkens2009}, which remains the same in the entanglement-assisted $3$-bit RAC \cite{Ghorai2018}. Crucially, the $3$-bit variant of PMRAC provides an advantage over both cases.

We demonstrate the optimal quantum success probability of the variant of PMRAC in Eq. (\ref{opts3}) self-tests the shared state between Alice and Bob, Bob's measurements and Alice's unitary operations which are encapsulated in the following two Theorems.

\begin{thm}\label{lemma1} The optimal quantum success probability $\mathcal{S}_Q^{opt}$ of the variant of $3$-bit PMRAC certifies the following statements.
\begin{eqnarray}
&&(i) \ \text{The states in set $\{\rho_{x|\oplus_{2} x=0}\}$ and $\{\rho_{x|\oplus_{2} x=1}\}$ are mutually }\nonumber\\ 
&& \quad \text{orthogonal satisfying $\sum_{x}\rho_{x|\oplus_{2} x=0}=\sum_{x}\rho_{x|\oplus_{2} x=1}=\openone_{2}\otimes\openone_{2} $.}\nonumber\\
\nonumber
    \nonumber
     && (ii) \ \Tr[\rho_{x}  \ \rho_{x^{\prime}}]= 
    \begin{cases}
    0 \quad \forall x\in\{0,1\}^3, x^{\prime}=\Bar{x}\\
    \frac{1}{3} \quad otherwise
    \end{cases}\label{Trrho}\\
    \nonumber
    && (iii) \text{Bob's observables}  \ B_y=\sqrt{\frac{3}{8}}\qty(N_y-M_y) \quad \forall y\in[3]\label{obBi} 
\end{eqnarray}
where $M_y$ and $N_y$ are two-qubit operators defined in Eq.~(\ref{tila1}) and (\ref{tila2}) respectively. \end{thm}

Note that statement (i) has already been proved above. A detailed proof of statements (ii) and (iii) are placed in Appendix~\ref {appsos}.

In Theorem \ref{lemma1}, the relations between $\rho_{x}$ and Bob's observable are certified. 
  Next, we demonstrate our main result of self-testing of Alice's unitary operations $U_x$ with $x\in \{0,1\}^{3}$ in Theorem \ref{thm2}.

\begin{thm}\label{thm2}
For optimal quantum success probability $\mathcal{S}^{opt}_Q$ of the $3$-bit variant of PMRAC, the following statements hold.
\begin{enumerate}[(i)]
    \item The initial shared state between Alice and Bob is maximally entangled state in the form
    \begin{eqnarray}
\label{ent}
    \rho=\frac{1}{4}\left(\openone_4-Q_1\otimes Q_2-P_1\otimes P_2-R_1\otimes R_2\right)
\end{eqnarray}
where $\{P_1, Q_1, R_1\}$ and $\{P_2, Q_2,R_2\}$ are the sets of mutually anti-commuting qubit operators. 
    \item Alice's unitary operators are $U_{000}=\openone_2,U_{011}= i \ Q_1,U_{101}=-i \ P_1$, and $U_{110}=-i \ R_1$, are mutually anticommuting, and a grand unitary $U_G=\frac{-\openone_2+U_{011}+U_{101}}{\sqrt{3}}$ exists, which gives \begin{eqnarray}
    U^{\dagger}_G\{\rho_{000},\rho_{011},\rho_{110},\rho_{101}\} U_G \rightarrow \{\rho_{001},\rho_{010},\rho_{100},\rho_{111}\}\nonumber
    \end{eqnarray}
   \end{enumerate}
     
\end{thm}

\begin{proof}

     The optimal quantum value requires the set  $\{\rho_{000},\rho_{011},\rho_{101},\rho_{110}\}$ to be mutually orthogonal pure states and forms a complete basis as explicitly proved in statement (i) of Theorem 1. Since, only Alice can apply the quantum operation on the shared state $\rho$, to produce four mutually orthogonal two-qubit states, the state $\rho$ has to be maximally entangled. For instance, commencing with $\ket{\psi}=a\ket{00} + b\ket{11}$ and applying local unitary operations on one side of the state, a basis of four orthogonal states can be produced only if $|a| = |b|$ holds true. This proves the statement (i) and a maximally entangled two-qubit state can always be written in the form in Eq. (\ref{ent}). For the proof of statement (ii) we proceed as follows.
 
 Let us now explicitly write the states in Eqs.~(\ref{tila1}) and (\ref{tila2}) in the following form.
\begin{eqnarray}\label{as}
 &&\rho_{000} = \frac{\openone_4 - M_1 - M_2 - M_3}{4}, 
  \rho_{011} = \frac{\openone_4 - M_1 + M_2 + M_3}{4} \quad\label{asa} \\
  &&\rho_{101} = \frac{\openone_4 + M_1 - M_2 + M_3}{4},
  \rho_{110} = \frac{\openone_4 + M_1 + M_2 - M_3}{4}\nonumber\\
  &&\rho_{001} = \frac{\openone_4 + N_1 + N_2 - N_3}{4},
  \rho_{010} = \frac{\openone_4 + N_1 - N_2 + N_3}{4}  \label{asb}\\
  &&\rho_{100}= \frac{\openone_4 - N_1 + N_2 + N_3}{4},
  \rho_{111} = \frac{\openone_4 - N_1 - N_2 - N_3}{4} \nonumber
\end{eqnarray}

The two-qubit operators $M_1, M_2$ and $M_3$ in Eq. (\ref{asa}) are mutually commuting and by comparing with Eq. (\ref{ent}) the general form can be written as
\begin{eqnarray}\label{QPR}
    M_1=Q_1\otimes Q_2, \quad M_2=P_1\otimes P_2, \quad M_3=R_1\otimes R_2
\end{eqnarray}
where $[M_y,M_{y^\prime}]_{y\neq {y^\prime}}=0,\forall y,{y^\prime}\in[3]$ with $M_1=-M_2 M_3$. The qubit operators $P_1$, $Q_1$ and $R_1$ are mutually anti-commuting and the same holds for $P_2$, $Q_2$ and $R_2$. 

Without loss of generality, we consider that Alice and Bob initially share the state $\rho\equiv\rho_{000}$, so that $U_{000}=\openone_2$. Then 
 \begin{eqnarray}\label{Uijk}
    \rho_{x_1 x_2  x_3}=(U^\dagger_{x_1 x_2  x_3}\otimes\openone_2) \ \rho_{000} \ (U_{x_1 x_2  x_3}\otimes\openone_2) \quad \forall x_1,x_2,x_3\in\{0,1\}\nonumber\\
\end{eqnarray}

We first consider the generation of $\rho_{011}$ from $\rho_{000}$, so that $\rho_{011=}U^{\dagger}_{011}\rho_{000}U_{011}$. This requires $U_{011}^{\dagger}Q_1 U_{011}=Q_1$, $U_{011}^{\dagger} P_1 U_{011}=-P_1$ and $U_{011}^{\dagger} R_1 U_{011}=-R_1$. This in turn gives $U_{011}=i\ Q_1$. A similar derivation gives $U_{101}=-i \ P_1$ and $U_{110}=-i \ R_1$.

Now, to generate the states in Eq. (\ref{asb}), we simply need to find a grand unitary $U_G$ that transforms the basis in Eq. (\ref{asa}) to Eq. (\ref{asb}), i.e., 
\begin{eqnarray}
U^{\dagger}_G\{\rho_{000},\rho_{011},\rho_{110},\rho_{101}\} U_G \rightarrow \{\rho_{001},\rho_{010},\rho_{100},\rho_{111}\}
\end{eqnarray}
which require 
\begin{eqnarray}
\nonumber
    &&(U^\dagger_{G}\otimes\openone_2)M_1(U_{G}\otimes\openone_2)=-N_1\label{U1A}\\&&(U^\dagger_{G}\otimes\openone_2)M_2(U_{G}\otimes\openone_2)=-N_2\label{U1B}\\
    \nonumber
&&(U^\dagger_{G}\otimes\openone_2)M_3(U_{G}\otimes\openone_2)=N_3\label{U1C}
\end{eqnarray}

As proved in Appendix \ref{appsos}, the following relations $[M_y,M_{y^\prime}]=[N_y,N_{y^\prime}]=0,\forall y\neq {y^\prime}\in[3]$, and
\begin{eqnarray}\label{as1} \Tr[M_y N_y]=-\frac{4}{3},\forall y\in[3]
\end{eqnarray}
hold with $N_y^2=\openone_4$. Without loss of generality, we can take 

\begin{eqnarray}\label{def1}
    N_y=-\frac{M_y}{3} + a_y \frac{2\sqrt{2} M^{\prime}_y}{3},\forall y\in[3]
\end{eqnarray} 
with $\{M_y, M^{\prime}_{y}\}=0$ and $a_y=\pm 1$. Note that the unitary operator acts only on Alice's local subsystem, that is, on the first Hilbert space of the two-qubit observables. Consider the example of $N_1$ in Eq. (\ref{def1}). Since $M_1=Q_1\otimes Q_2$, then $M_1'$ can always be written as $M_1'=\frac{P_1+R_1}{\sqrt{2}}\otimes Q_2$, and $a_1=1$ satisfying $\{M_1, M_1'\}=0$.

Similarly, from $\{M_2,M_2'\}=0$ and $\{M_3,M_3'\}=0$ we can write $M_2'=\frac{Q_1+ b_2 R_1}{\sqrt{2}}\otimes P_2$ and $M_3'=\frac{Q_1+ b_3 P_1}{\sqrt{2}}\otimes R_2$ respectively, where $b_2, b_3 =\pm 1$. Note that for satisfying $[N_y,N_{y^\prime}]=0,\forall y\neq {y^\prime}\in [3]$, only possible solution is $a_y=b_y=+1$.

By incorporating this in Eq.~(\ref{U1A}) we get
\begin{eqnarray}\nonumber
    &&U^\dagger_{G}\ Q_1\ U_{G}\otimes Q_2 =\frac{Q_1- 2(P_1+R_1)}{3}\otimes Q_2\label{U1a}\\
   && U^\dagger_{G} \ P_1 \ U_{G}\otimes P_2 =\frac{P_1-2( Q_1+R_1)}{3}\otimes P_2\label{U1b} \\
    &&U^\dagger_{G} \ R_1 \ U_{G} \otimes R_2 =\frac{-R_1+2 (Q_1+ P_1)}{3}\otimes R_2\label{U1c}\nonumber
\end{eqnarray}

 Noting that $Q_1, P_1$ and $R_1$ are mutually anticommuting qubit observables, the general form of $U_{G}$ can be written as 
\begin{equation}\label{u001m}
U_{G}=p\openone_2+q Q_1+r P_1+s R_1
\end{equation}where $\{||p||,||q||,||r||,||s||\}\in[0,1]$. Putting Eq.~(\ref{U1b}) into Eq. (\ref{u001m}), we get $p=-\frac{1}{\sqrt{3}},q=\frac{i}{\sqrt{3}},r=-\frac{i}{\sqrt{3}}$ and $ s=0$. Hence, the grand unitary is derived as follows
\begin{eqnarray}
    U_{G}=\frac{-\openone_2+i \ Q_1-i \ P_1}{\sqrt{3}}\equiv \frac{-\openone_2+U_{011}+U_{101}}{\sqrt{3}}
\end{eqnarray}

\end{proof}

\textit{An example of state and measurements:} We provide an example as follows.
$M_1=\sigma_x\otimes \sigma_x$, $M_2=\sigma_y\otimes \sigma_y$, $M_3=\sigma_z\otimes \sigma_z$ and $N_1=(-\sigma_x\otimes\sigma_x+2\sigma_y\otimes\sigma_x+2\sigma_z\otimes\sigma_x)/3$, $N_2=(-\sigma_y\otimes\sigma_y+2 \sigma_x\otimes\sigma_y+ 2\sigma_z\otimes\sigma_y)/3$, $N_3=(-\sigma_z\otimes\sigma_z+2\sigma_x\otimes\sigma_z+2\sigma_y\otimes\sigma_z)/3$, and Bob's local observables are
$B_1= \frac{-2 \sigma_x+\sigma_y+\sigma_z}{\sqrt{6}}\otimes \sigma_x$, $B_2= \frac{-2 \sigma_y+\sigma_x+\sigma_z}{\sqrt{6}}\otimes \sigma_y$, $B_3= \frac{-2 \sigma_z+\sigma_x+\sigma_y}{\sqrt{6}}\otimes \sigma_z$. This fixes $
    U_{011}=i \ \sigma_x,\;
      U_{101}=-i \ \sigma_y,\;
      U_{110}=-i \ \sigma_z$ and the grand unitary  $U_G=\frac{-\openone_2+i \ \sigma_x-i \ \sigma_y}{\sqrt{3}} $.

\emph{Summary and Discussion:-}
In sum, we demonstrated a novel self-testing protocol to certify unitary operations. This is based on the optimal quantum advantage a communication game in a variant of prepare-measure scenario. In particular, we considered a $3$-bit variant of quantum PMRAC where Alice and Bob share a prior two-qubit state, and to encode the message, Alice applies unitary operations on her part of the system and sends her system to Bob. Upon receiving Alice's system, Bob performs the measurement on the whole system to decode the message. 
By developing an elegant analytical approach, we derived the optimal quantum success probability of the variant of $3$-bit PMRAC which self-tests the shared state between Alice and Bob, Alice's unitary operations $U_{x}$ with $x\in\{0,1\}^3$ and Bob's measurements $B_{y}$ with $y\in[3]$. 

We remark that the certification of unitary operations possesses an independent interest elsewhere, for example, in quantum computing. An unitary operation is a reversible and norm-preserving evolution of quantum states that forms the basis for quantum circuits, as well as delegated quantum computing \cite{Shor1994,Grover1996, Fitzsimons2017}. Our work may be seen as an important step toward SDI validation of quantum circuits.

Our work should spur several follow-up works. An immediate extension could be to generalize our self-testing protocol for $n$-bit quantum PMRAC in the variant of prepare-measure scenario by considering higher-dimensional states. The quantum advantage can depend on the dimension of the shared  state, and therefore the $n$-bit quantum PMRAC can serve as a dimension witness. The fully analytical approach we developed to derive the optimal quantum success probability is so elegant that it can be used straightforwardly for the $n$-bit quantum PMRAC. Our approach can also be used for other communication games in the variant of prepare-measure scenario. This could be an exciting avenue for future research and thus calls for further study.

\emph{Acknowledgments:-} 
RP acknowledges the financial support from the Council of Scientific and Industrial Research (CSIR, 09/1001(12429)/2021-EMR-I), Government of India. PR acknowledges the support by KIAS individual Grant No. QP100601 at the Korea Institute for Advanced Study and by the local hospitality from the grant IITH/SG160 of IIT Hyderabad, India. AKP acknowledges support from Research Grant No. SERB/CRG/2021/004258, Government of India. 

\appendix

\section{Derivation of statements (ii) and (iii) of Theorem \ref{lemma1}} \label{appsos}
From Eqs. (\ref{tila1}) and (\ref{tila2}) in the main text, we get
\begin{eqnarray}
\Tr[M_1 N_1]=-4&+&2\Big(\Tr[\rho_{000}\rho_{100}]+\Tr[\rho_{011}\rho_{111}]\nonumber \\
&+&\Tr[\rho_{101}\rho_{001}]+\Tr[\rho_{110}\rho_{010}]\Big)\\ \label{m1n1}
        \Tr[M_2 N_2]=-4&+&2\Big(\Tr[\rho_{000}\rho_{010}]+\Tr[\rho_{011}\rho_{001}]\nonumber  \\
        &+&\Tr[\rho_{101}\rho_{111}]+\Tr[\rho_{110}\rho_{100}]\Big)\\\label{m2n2}
        \Tr[M_3 N_3]=-4&+&2\Big(\Tr[\rho_{000}\rho_{001}]+\Tr[\rho_{011}\rho_{010}]\nonumber \\
        &&+\Tr[\rho_{101}\rho_{100}]+\Tr[\rho_{110}\rho_{111}]\Big)\label{m3n3}
\end{eqnarray}
Note that the maximization requires $\omega_{1}=\omega_{2}=\omega_{3}$. Thus, $\omega_{1}=\omega_{2}$ implies $\Tr[M_1 N_1]=\Tr[M_2 N_2]$, which gives,
\begin{eqnarray}\label{w12}
  && \Tr[\rho_{000}(\rho_{100}-\rho_{010})]-\Tr[\rho_{110}(\rho_{100}-\rho_{010} )]\nonumber \\
   &&+\Tr[\rho_{011} (\rho_{111}-\rho_{001} )]-\Tr[\rho_{101} (\rho_{111}-\rho_{001} )]=0 
\end{eqnarray}
Similarly, $\omega_{1}=\omega_{3}$ gives,
\begin{eqnarray}\label{w13}
    &&\Tr[\rho_{000}(\rho_{100}-\rho_{001} )]-\Tr[\rho_{101} (\rho_{100}-\rho_{001} )]\nonumber \\
    &&+\Tr[\rho_{011}(\rho_{111}-\rho_{010} )]-\Tr[\rho_{110}(\rho_{111}-\rho_{010} )]=0 
\end{eqnarray}
and $\omega_{2}=\omega_{3}$ gives,
\begin{eqnarray}\label{w23}
        &&\Tr[\rho_{000} (\rho_{010}-\rho_{001} )]-\Tr[\rho_{011} (\rho_{010}-\rho_{001} )]\nonumber \\
        &&+\Tr[\rho_{101}(\rho_{111}-\rho_{100} )]-\Tr[\rho_{110}(\rho_{111}-\rho_{100} )]=0 
\end{eqnarray}
Further, using the relation $Tr[M_1 N_1]+Tr[M_2 N_2]+Tr[M_3 N_3]=-4$, we get 
\begin{equation}
    Tr[M_1 N_1]=Tr[M_2 N_2]=Tr[M_3 N_3]=-\frac{4}{3}
\end{equation}

Hence, from Eqs. (\ref{m1n1}), (\ref{m2n2}) and (\ref{m3n3}) we can write the following relations.
\begin{eqnarray}\label{2/3}\nonumber
    \Tr[\rho_{000}\rho_{100}]+\Tr[\rho_{011}\rho_{111}]+\Tr[\rho_{101}\rho_{001}]+\Tr[\rho_{110}\rho_{010}]&=&\frac{4}{3}\nonumber\\
   \Tr[\rho_{000}\rho_{010}]+\Tr[\rho_{011}\rho_{001}]+\Tr[\rho_{101}\rho_{111}]+\Tr[\rho_{110}\rho_{100}] &=&\frac{4}{3}\nonumber\\
   Tr[\rho_{000}\rho_{001}]+\Tr[\rho_{011}\rho_{010}]+\Tr[\rho_{101}\rho_{100}]+\Tr[\rho_{110}\rho_{111}] &=&\frac{4}{3}\nonumber    \\
\end{eqnarray}
By writing $\Tr[\rho_{000} \ \openone_2\otimes \openone_2]= 1$, we get
\begin{eqnarray}\label{p0}
    &&\Tr[\rho_{000}\qty(\rho_{001}+\rho_{010}+\rho_{100}+\rho_{111})]=1\nonumber\\
  &&\Rightarrow   \Tr[\rho_{000}(\rho_{001}+\rho_{010}+\rho_{100})]=1, \ \qty(\text{as} \Tr[\rho_{000}  \rho_{111}]=0)\quad\quad
\end{eqnarray}
Similarly, we can write 
\begin{eqnarray}
    &&\Tr[\rho_{011}(\rho_{001}+\rho_{010}+\rho_{111})]=
    \Tr[\rho_{101}(\rho_{001}+\rho_{100}+\rho_{111})]\nonumber\\
    &=&
    \Tr[\rho_{110}(\rho_{010}+\rho_{100}+\rho_{111})]=\Tr[\rho_{010}(\rho_{000}+\rho_{011}+\rho_{110})]\nonumber\\ \label{p10}
    &=&\Tr[\rho_{100}(\rho_{000}+\rho_{101}+\rho_{110})]=\Tr[\rho_{111}(\rho_{011}+\rho_{101}+\rho_{110})]\nonumber \\
    &=&\Tr[\rho_{001}(\rho_{000}+\rho_{011}+\rho_{101})]=1
    \label{p11}
\end{eqnarray}
Note that Eqs.~(\ref{2/3}), (\ref{p0}) and (\ref{p11}) are simultaneously satisfied only when the following conditions are met.
\begin{eqnarray}\label{c3}
   &&\Tr[\rho_{000} \ \rho_{001}]=\Tr[\rho_{000} \ \rho_{010}]=\Tr[\rho_{000} \ \rho_{100}]=\Tr[\rho_{011} \ \rho_{001}]\nonumber\\
   &=&\Tr[\rho_{011} \ \rho_{010}]=\Tr[\rho_{011} \ \rho_{111}]=\Tr[\rho_{101} \ \rho_{001}]=\Tr[\rho_{101} \ \rho_{100}]\nonumber\\
    &=&\Tr[\rho_{101} \ \rho_{111}]= \Tr[\rho_{110} \ \rho_{010}]=\Tr[\rho_{110} \ \rho_{100}]=\Tr[\rho_{110} \ \rho_{111}]
    \nonumber\\
  &=&\frac{1}{3}
\end{eqnarray}

In a compact form, the conditions in Eqs. (\ref{ovzero}) in the main text and (\ref{c3}) can be written as
\begin{eqnarray}
   \ \Tr[\rho_{x}  \ \rho_{x^{\prime}}]= 
    \begin{cases}
    0 \quad \forall x\in\{0,1\}^3, x^{\prime}=\Bar{x}\\
    \frac{1}{3} \quad \text{otherwise}
    \end{cases}
\end{eqnarray}
which is statement (ii) in Theorem \ref{lemma1}.
With the help of Eq.~(\ref{c3}) we prove  that $\omega_y=\sqrt{\frac{8}{3}}, \forall y\in[3]$. Hence, to obtain the maximum quantum success probability, Bob's observables should satisfy $B_y=\mathscr{M}_y=\sqrt{\frac{3}{8}}\qty(N_y-M_y),\forall y\in [3]$, which is the statement (iii) in Theorem \ref{lemma1}.

\bibliographystyle{apsrev4-2} 
\bibliography{references1}

\end{document}